\documentclass[letterpaper, 10 pt, conference]{ieeeconf}  

\IEEEoverridecommandlockouts                              
\usepackage{graphics} 
\usepackage{epsfig} 
\usepackage{mathptmx} 
\usepackage{times} 

\usepackage{amsthm}
\usepackage{amsmath}
\usepackage{amssymb}  
\usepackage{amsfonts}
\usepackage{algpseudocode}
\usepackage{algorithm}
\usepackage{xcolor}
\usepackage{balance}
\usepackage{microtype}

\usepackage{cite}

\usepackage{graphicx}
\usepackage{float,epsf,caption,subcaption}
\usepackage{multirow}

\newtheorem{theorem}{Theorem}

\newtheorem{corollary}{Corollary}
\newtheorem{proposition}{Proposition}
\theoremstyle{plain}
\newtheorem{lemma}{Lemma}
\theoremstyle{definition}
\newtheorem{definition}{Definition}
\newtheorem{example}{Example}
\theoremstyle{remark}
\newtheorem{remark}{Remark}

\newcommand{\R}{\mathbb{R}}
\newcommand{\END}{\hfill $\blacksquare$}

\makeatletter
\def\endthebibliography{%
	\def\@noitemerr{\@latex@warning{Empty 'thebibliography' environment}}%
	\endlist
}
\makeatother
\title{\LARGE \bf
	Automaton-based Implicit Controlled Invariant Set Computation for Discrete-Time Linear Systems
}

\author{Zexiang Liu\quad\;\; Tzanis Anevlavis\quad\;\; Necmiye Ozay\quad\;\; Paulo Tabuada

\thanks{Z. Liu and N. Ozay are with the Dept. of Electrical Engineering and Computer Science, Univ. of Michigan, Ann Arbor, MI 48109, USA
        {\tt\small \{zexiang, necmiye\} @umich.edu}. T. Anevlavis and P. Tabuada are with the UCLA Electrical and Computer Engineering Department,
	Los Angeles, CA 90095
        {\tt\small \{janis10, tabuada\} @ucla.edu}. 
}
\thanks{This work was partially supported by  the NSF grants CNS-1931982 and CCF-1918123 and the CONIX Research Center, one of six centers in JUMP, a Semiconductor Research Corporation (SRC) program sponsored by DARPA.}
}
\begin{document}
	
	\maketitle
	\thispagestyle{empty}
	\pagestyle{empty}

	\begin{abstract}
	In this paper, we derive closed-form expressions for implicit controlled invariant sets for discrete-time controllable linear systems with measurable disturbances. In particular, a disturbance-reactive (or disturbance feedback) controller in the form of a  parameterized finite automaton is considered. We show that, for a class of automata, the robust positively invariant sets of the corresponding closed-loop systems can be expressed by a set of linear inequality constraints in the joint space of system states and controller parameters. This leads to an implicit representation of the invariant set in a lifted space. We further show how the same parameterization can be used to compute invariant sets when the disturbance is not available for measurement. 
	\end{abstract}



\section{Introduction}

When tasked with synthesizing a controller in order to ensure safety of a plant under uncertainties, the objective is to indefinitely keep the state of the plant within a set of safe states. A natural solution to this task is to initialize the state in a  Robust Controlled Invariant Set (RCIS) within the set of safe states. RCISs have the property that any trajectory starting within an RCIS can always be forced to remain inside the RCIS and, therefore, inside the set of safe states.
Consequently, RCISs are at the core of safety-critical applications. 

Since the conception of the standard method for computing the Maximal RCIS  of discrete-time systems \cite{bertsekas1972infreach, nilsson2017correct}, is known to suffer from poor scaling with the system's dimension and no guarantees of termination, different approaches attempted to alleviate these drawbacks. In the case of bounded disturbances, when the set of safe states are polytopes, \cite{rungger2017rcislinsys} computes inner and outer approximations of RCISs for linear systems with guarantees on finite-time termination. For the same system class, a different line of works \cite{tahir2014low, gupta2018full,michel2018invariant} approximate the Maximal/Minimal RCIS by first closing the loop with a linear state feedback law and then computing the Robust Positively Invariant Set (RPIS) of the closed-loop system. This group of methods are typically very conservative since only linear state feedback controllers are considered. 

Finally, several recent methods   \cite{anevlavis2019cis2m,anevlavis2020simple,anevlavis2021enhanced, wintenberg2020implicit}, including previous works from the authors, develop approaches for constructing implicit RCIS in closed form, represented as polytopes in a high-dimensional space whose projection onto the original state space is an RCIS. By avoiding  computing RCISs explicitly, those methods can work for systems with higher dimensions. It is indeed the case in many practical applications, such as model predictive control and supervision control, that knowledge of the explicit RCIS is not required and the implicit representation suffices\cite{anevlavis2021controlled,liu2021safe}. 

Inspired by the recent progress on implicit RCISs, in this work we propose a novel approach to compute implicit RCISs for discrete-time linear systems. In addition, the aforementioned works consider non-measurable disturbances only, however, in many safety-critical applications, incoming disturbances can be measured in ahead\cite{xu2019design} and, hence, are considered measurable \cite{nilsson2017correct}. Thus, unlike the existing works, we develop a method that works for both measurable and non-measurable disturbances, achieved by introducing an automaton-based controller whose input is exactly the measurable disturbance.
More specifically, our contributions are as follows:   \\
\hspace*{2mm} 1) We propose an automaton-based method for computing \emph{implicit RCISs} for a class of linear systems that contains the class of controllable linear systems, with measurable disturbances. \\
\hspace*{2mm} 2) We derive conditions on the structure of the automaton such that the implicit RCIS is computed in closed-form.\\
\hspace*{2mm} 3) We present a generic connection between measurable and non-measurable disturbances, enabling the proposed method to work with systems with non-measurable disturbances. 

In addition to the above, we demonstrate the practicality of the implicit RCIS in the task of supervision control for the lane keeping problem. The goal is to modify nominal control inputs, as needed, to keep the system’s trajectory within a set of safe states. We show that this is achieved by solving an optimization problem using the implicit RCIS. 

The paper is organized as follows: In Section~\ref{sec:setup}, the problem is mathematically set up, along with the essential definitions and assumptions. Section \ref{sec:framework} lays down the ideas for computing an implicit RCIS for systems with measurable disturbances. Subsequently, Section \ref{sec:closedform} investigates when the implicit RCIS can be computed in closed-form, while Section \ref{sec:bridge} connects the proposed method to the case of non-measurable disturbances. Section~\ref{sec:compEval} provides a computational evaluation of the proposed method. Finally, conclusion is found in Section~\ref{sec:concl}.
To keep a streamlined presentation, the proofs of all theorems are found in Appendix.

\noindent \textbf{Notation:} The Minkowski sum of two sets $A$ and $B$ is denoted by $A+B = \{a+b \mid a\in A, b\in B\} $. For a singleton $\{x\} $, we write the Minkowski sum $\{x\} + B $ as $x+B$. We denote the convex hull of a set $A$ by $ \mathbf{CH}(A)$. For a sequence $(d(t))_{t=0}^{N}$, we denote the finite subsequence $(d(t))_{t=a}^{b}$ by $d(a:b)$ with $0\leq a\leq b \leq N$. We denote the projection of a set $C$ in $\R^{n+m}$ onto the first $n$ coordinates by $\textrm{Proj}_{1:n}(C)$ in $\R^n$. For two vectors $x\in \R^n$ and $y\in \R^m$, $(x,y)$ denotes the concatenated vector in $\R^{n+m}$.

\section{Preliminaries}
\label{sec:setup}
In this work, we consider a discrete-time linear system $\Sigma$:
\begin{align} 
\label{eq:dtls} 
	\Sigma:	x(t+1) = A x(t) + Bu(t) + d(t), 
\end{align}
where $x\in \R^{n}$ is the state of $\Sigma$, $u \in \R^{m}$ is the input, and \mbox{$d\in D \subseteq \R^{n}$} is a disturbance term. The set $D$ contains all possible values of $d$. The disturbance $d$ is \emph{measurable} if $u(t)$ can be determined based on the measurement $d(t)$; otherwise, $d$ is \emph{non-measurable}. 
\begin{definition}[Safe set]
\label{def:safeset}
Let $S \subseteq \R^{n+m}$ be the set of desired state and input pairs $(x,u)$, called the \emph{safe set} of the system. That is, we want $(x(t),u(t))\in S$ for all $t\geq 0$. 
\end{definition}

The difference between measurable and non-measurable disturbances is reflected in the following definitions of a Robust Controlled Invariant Set (RCIS). 
 
\begin{definition}[RCIS with non-measurable disturbance]
\label{def:cisnonmeas}
Consider a \emph{non-measurable} disturbance $d$. A set $C \subseteq \R^{n}$ is an RCIS for the system $\Sigma$ within the safe set $S$ if:
\begin{align}
\label{eq:cisnonmeas}
\hspace{-1.8mm}	\forall x\in C, \exists u \text{ such that } (x,u)\in S, Ax+Bu+d \in C,\forall d\in D.
\end{align}
\end{definition}

\begin{definition}[RCIS with measurable disturbance]
\label{def:cismeas}
For the disturbance $d$ being \emph{measurable}, a set $C \subseteq  \R^{n}$ is an RCIS for the system $\Sigma$ within the safe set $S$ if: 
\begin{align}
\label{eq:cismeas}
\hspace{-1.8mm}	\forall x\in C, \forall d\in D,\exists u \text{ such that } (x,u)\in S, Ax+Bu+d \in C.
\end{align}
\end{definition}
Note that the order of the quantifiers $\forall d \in D$ and $\exists u$ is swapped in the  definitions above, which means an RCIS $C$ with respect to a non-measurable disturbance is guaranteed to be an RCIS for the same system with measurable disturbance but not vice versa. Finally, we call a set $C_{max}$ the \emph{Maximal RCIS} within $S$ if it is controlled invariant and contains every RCIS in $S$. 

Next, consider an autonomous system $\Sigma_{a}$:
\begin{align}
\Sigma_{a}: x(t+1) = f(x(t),d(t))
\end{align}
with state $x\in \R^{n}$ and a disturbance term $d\in D$. 

\begin{definition}[Robust Positively Invariant Set (RPIS)]
	\label{def:fis}
	A set $C \subseteq S$ is an \emph{RPIS} of the autonomous system $\Sigma_{a}$ within $S$ if:
	\begin{align}
	\label{eq:fis}
	\forall x\in C \text{ we have that } f(x,d)\in C,  \forall d\in D.  
	\end{align}
	A set $C_{max}$ is the \emph{Maximal RPIS} within $S$ if it is positively invariant and contains every robust positively invariant set in $S$. 
\end{definition}
Notice that, compared to systems with control inputs, there is no concept of measurable or non-measurable disturbances for autonomous systems. 

\begin{definition}[Reachable set for autonomous systems]
	\label{def:reachautonsys}
	Let $\Sigma_a$ be an autonomous system. 
	The \emph{reachable set} $\mathcal{R}(\Sigma_{a}, x_0)$ of $\Sigma_{a}$ from state $x_0$ is the set of all possible states that the system may visit. 
	Formally, $x\in \mathcal{R}(\Sigma_{a},x_0)$ if and only if there exists a trajectory $(x(t))_{t=0}^{k}$, for some $k\geq0$, of $\Sigma_{a}$ under disturbance sequence $(d(t))_{t=0}^{k-1}\in D^{k}$ with $x(0) = x_0 $ and $x(k) = x$.
\end{definition}

\begin{proposition} \label{thm:RInv} 
	The set $C = \{x\in \R^{n} \mid \mathcal{R}(\Sigma_{a},x) \subseteq S\} $, i.e., the set of states whose corresponding reachable set is contained in the safe set $S$, is the Maximal RPIS for $\Sigma_{a}$ within $S$.	
\end{proposition}

\subsection{Problem Setup}
In the first part of this work, we focus on the computation of RCISs for systems with measurable disturbances. More specifically, we propose a method that computes the desired implicit representation for an RCIS in closed-form based on the following assumptions. In the second part, we extend our method to compute RCISs for non-measurable disturbances.

\noindent \textbf{Assumption 1:}  The matrix $A$ is nilpotent. That is, there exists a non-negative integer $h \leq n$ such that $A^{h}= 0$.

\begin{remark}
	Assumption 1 is satisfied by any controllable system, as there always exists a feedback gain $ K $ such that $ A+BK $ is nilpotent. Thus, any controllable system satisfies Assumption 1 by pre-feedbacking the system with $ u = Kx +v $ and taking  $ v $ as the new control input\footnote{Accordingly, given the safe set $S_{xu}$ for $(x,u)$, the safe set $S_{xv}$ for the new state-input pair $(x,v)$ should be $S_{xv} = \{(x,v)\mid (x,Kx+v)\in S_{xu}\}$.} \cite[Ch.3]{Antsaklis2006LinSys}.
\end{remark}

\noindent \textbf{Assumption 2:} The safe set $S \subset \R^{n+m}$ and the disturbance set $D \subset \R^{n}$ are both polytopes. 

The next theorem shows that to compute an RCIS for systems with a measurable disturbance, we only need to consider the (finite) vertices of $D$. 

\begin{proposition} \label{thm:finite_D} 
	Under Assumption 2, consider a system $\Sigma$ 
	with a measurable disturbance $d\in D$, and let $D_{v}$ be the set of vertices of $D$. 
	Let $\Sigma'$ be the same system as $\Sigma$ but with the measurable disturbance $d\in D_v$.
	Then, a convex set $C$ is an RCIS for $\Sigma$ if and only if $C$ is an RCIS for $\Sigma'$.
\end{proposition}

According to Proposition~\ref{thm:finite_D}, given any polytopic disturbance set $D$, we can substitute $D$ by the finite set $D_{v}$ without loss of generality. Thus, for the remaining of this paper, we directly assume that $D$ is a finite set, as stated below.

\noindent \textbf{Assumption 3:} {The disturbance set $D \subset \R^{n}$ is given as a finite set of vertices.}

\noindent \textbf{Problem 1:} For a linear system $\Sigma$, 
a safe set $S$, and a disturbance set $D$ satisfying Assumptions 1, 2 and 3, compute a convex implicit RCIS  of $\Sigma$ within $S$ in closed-form.

\section{Controlled Invariant Set Computation Framework} 
\label{sec:framework}  
The Maximal RCIS for a system in \eqref{eq:dtls} with measurable and/or non-measurable disturbances can be computed by a standard iterative method (\cite{nilsson2017correct,bertsekas1972infreach}), which is not guaranteed to terminate in finite time and does not scale to high-dimensional systems. To reduce the computation burden, many existing works close the loop with a linear feedback controller and then compute an RPIS of the closed-loop system as an under-approximation of the Maximal RCIS.  In this section, we extend this idea to a more general case: We close the loop with a parameterized nonlinear disturbance-feedback controller. Then, by computing an RPIS of an augmented closed-loop system, we search for the feasible initial states and controller parameters simultaneously such that the closed-loop trajectory satisfies the safety constraints.  Lastly, we retrieve an RCIS from the RPIS of the augmented system.

First, we want to determine an appropriate controller structure. We draw some inspirations from Definition \ref{def:cismeas}: Given any RCIS  $ C $ for $\Sigma$ under measurable disturbances, by definition, there exists a \emph{memoryless state-disturbance feedback controller} 
\mbox{$u(t) = k\left( x(t), d(t) \right)$} such that $ C $ is the Maximal RPIS of the closed-loop system. In other words, any RCIS,  including the Maximal RCIS, is the Maximal RPIS of a closed-loop system with respect to some memoryless state-disturbance controller. Thus, to minimize the conservativeness of the closed-loop RPIS, it is enough to consider the class of memoryless state-disturbance feedback controllers. Furthermore, it is well-known that any memoryless state-disturbance feedback controller is equivalent to a \emph{disturbance feedback controller with memory}, explained by the following example.
\begin{example} \label{ex:1}
	Consider a memoryless state-disturbance feedback controller $u(t) = \kappa(x(t), d(t))$. This controller can be equivalently expressed as:
	\begin{align}
	\begin{cases}
	s(t+1) &=  As(t)+ B \kappa(s(t),d(t)) + d(t),    \\
	u(t) &=  \kappa(s(t), d(t)),
	\end{cases}
	\end{align}
	with $s(0) =  x(0)$. That is, the internal dynamics of the controller forms a state estimator. \END
\end{example}

For above reasons, we consider a \emph{parameterized disturbance-feedback controller $\Sigma_c$ with memory}: 
\begin{align} \label{eqn:cont} 
	\Sigma_{c}:
	\begin{cases}
	    s(t+1) &= \mathcal{T}(s(t), d(t); \theta),  \\
	    u(t) &= o(s(t), d(t); \theta). 
	\end{cases}
\end{align}
In the above, $s$ is the internal state (memory) of the controller that distills useful information from the disturbance input $d\in D$, and $u$ is the output of the controller. The same $d$ and $u$ correspond to the disturbance and the control input of $\Sigma$ respectively. The state transition function $\mathcal{T}$ and the output function $o$ map the current state $s(t)$ and the disturbance input $d(t)$ into the next state $s(t+1)$ and the current output $u(t)$ respectively. Finally, $ \theta$ is a constant vector that parameterizes the state transition function $\mathcal{T}$ and the output function $o$. The value of $\theta$ can depend on the initial state $x_{0}$ of the system $\Sigma$, such as in Example \ref{ex:1}, $ \theta = x(0) $. In what follows, we assume that the functions $\mathcal{T}$ and $o$ in $\Sigma_{c}$ are known. We discuss how to select $\mathcal{T}$ and $o$ in the next section. 

Closing the loop of $\Sigma$ in \eqref{eq:dtls} with the controller in \eqref{eqn:cont}, we obtain the following closed-loop system augmented with the controller internal state $s$ and the constant vector $\theta$:
\begin{align}
\hspace{-1.5mm} 
   \Sigma_{cl}: \hspace*{-0.5mm}
      \begin{bmatrix}
		  x(t+1)\\
		  \theta(t+1)\\
		  s(t+1)
      \end{bmatrix} =  
	  \begin{bmatrix}
		  A x(t)  + B o(s(t),d(t);\theta(t)) + d(t)\\
		  \theta(t)\\
		  \mathcal{T}(s(t), d(t); \theta(t))
      \end{bmatrix}
      ,
\end{align}
with the augmented state $(x,\theta, s)$ and the disturbance $d\in D$. 

In the above augmented system, we can calculate feasible initial states $x_0$ and controller parameters $\theta$ simultaneously. 
The control input $u(t)$ of $\Sigma$ is equal to $o(x(t),d(t); \theta(t))$, as a function of the augmented state. Then, given the safe set $S$ of the system $\Sigma$, we define the safe set $S_{cl}$ for the closed-loop system $\Sigma_{cl}$ by: 
\begin{align}
  S_{cl} = \left\{ (x, \theta,s) \mid (x,o(s,d; \theta))\in S, \forall d \in D)\right\}.
\end{align}
The following theorem connects the problem of computing an RCIS for the system $\Sigma$ with the problem of computing an RPIS for the closed-loop system $\Sigma_{cl}$.
\begin{theorem} \label{thm:CH} 
	Let $C_{cl}$ be an RPIS for the closed-loop system $\Sigma_{cl}$ within the safe set $S_{cl}$. Then, the convex hull $\textbf{CH} (\textrm{Proj}_{1:n}(C_{cl}))$ of the projection of $C_{cl}$ onto the first $n$ coordinates is a convex RCIS for the system $\Sigma$ within the safe set $S$.  
\end{theorem}
According to Theorem~\ref{thm:CH}, we propose a novel framework for computing RCISs: Given a controller specified by functions $\mathcal{T}$ and $o$, we first compute the Maximal RPIS $C_{cl}$ of the augmented system $\Sigma_{cl}$ within $S_{cl}$, and then take the convex hull of the projection of $C_{cl}$.  

The size of the resulting RCIS depends on the choice of functions $\mathcal{T}$ and $o$. In Section \ref{sec:dfts}, we demonstrate two classes of functions $\mathcal{T}$ and $o$ that lead to larger RCISs as the number of parameters increases. In theory, there exist functions $\mathcal{T}$ and $o$ such that the resulting RCIS meets the Maximal RCIS. However, how to find $\mathcal{T}$ and $o$ achieving the Maximal RCIS is beyond the scope of this paper.  One significant advantage of this framework, explored in the next section, is that by carefully designing $\mathcal{T}$ and $o$,  this novel framework enables closed-form construction of RCISs. 

\section{Closed-form construction of \\ implicit robust controlled invariant sets}
\label{sec:closedform}
In the previous section, we presented a framework for
computing RCISs. Still, there are two main questions to be
addressed. First, can we compute the Maximal RPIS $C_{cl}$
efficiently? Second, in practice the convex hull computation
is expensive, can we avoid this operation? In this section, we show that by carefully designing $\mathcal{T}$ and $o$, we address both questions.

\subsection{Computing $C_{cl}$ in closed-form}
From Proposition~\ref{thm:RInv}, we have that:
\begin{align}
	C_{cl} = \left\{(x, \theta, s) \mid \mathcal{R}(\Sigma_{cl}, (x,\theta,s)) \subseteq S_{cl}\right\}. \label{eqn:inv_reach} 
\end{align}
According to \eqref{eqn:inv_reach}, if we express the reachable set $\mathcal{R}(\Sigma_{cl}, (x,\theta,s))$ in closed-form, then we obtain a closed-form expression of $C_{cl}$. The next theorem gives a sufficient condition for the reachable set $ \mathcal{R}(\Sigma_{cl})$ to admit a closed-form expression.

\begin{theorem} \label{thm:finite_reach} 
	Under Assumption 1 and 3,	if the state $s$ of the controller belongs to a finite set $Q$, then given any initial state  $(x, \theta, s)$, the reachable set $\mathcal{R}(\Sigma_{cl}, (x, \theta,s))$ is finite. Moreover, it can be expressed in closed-form. 
\end{theorem}
Note that Theorem~\ref{thm:finite_reach} is the only result that requires Assumption 1 in this work.  Given Theorem~\ref{thm:finite_reach}, we want to design the functions $\mathcal{T}$ and $o$ such that the internal state $s$ of the controller belongs to a finite set $Q$. Recall that by Proposition~\ref{thm:finite_D} and Assumption 3, the input $d$ of the controller also belongs to a finite set $D$. Thus, the controller $\Sigma_{c}$ is a system with finite states and inputs. In the literature, this type of system is called a \emph{mealy machine}, a special class of automata.

\begin{definition}[Mealy Machine]
\label{def:dfts}
A \emph{Mealy Machine} $\Sigma_{fts}$ is a quintuple $(Q,D, \mathcal{T}, \Theta, o)$, 
where:
\begin{itemize}
	\item $Q$ is a finite set of discrete states;
	\item $D$ is a finite set of actions;
	\item $\mathcal{T}: Q\times D \rightarrow Q$ is the state transition function that maps each state-action pair to the next state;
	\item $\Theta$ is a finite set of outputs;
	\item $o: Q\times D \rightarrow	\Theta$ is the output function that maps each state-action pair to an element in the set $ \Theta$.
\end{itemize}    
\end{definition}
With slightly abusing notations, we denote both the action set of a mealy machine and the disturbance set of $\Sigma$ by $D$, since in this work we only consider the disturbance set $D$ as the action set.  The transition function $\mathcal{T}$ and the output function $o$ are designed by the user. 
 We  parameterize  the output set $ \Theta $ of a mealy machine by a parameter vector $ \theta $:  Suppose that  $ \Theta = \{u_i\}_{i=1}^L $, where each $ u_i $ is a vector of variables in $ \R^m $. Given a vector $\theta = ( \overline{u}_{1}, \cdots, \overline{u}_{L})\in \R^{mL} $, we define a parameterized output function  $o(s,d;  \theta): Q\times D \rightarrow \R^{m}$ such that $ o(s,d; \theta)= \overline{u}_{i}$ for $o(s,d)=u_{i}$. A simple example of a mealy machine and the parameterized output function is shown below.
\begin{example} \label{ex:toy_dfts} 
	Let $D = \{ d_1, d_2\}\subset \R^{n}$, and \mbox{$Q = \{s_1, s_2\}$}.  The state transition function $\mathcal{T}$ is shown in Fig.~\ref{fig:toy_fts}. The output set is $\Theta = \{ u_1, u_2\}$.  The output function is \mbox{$o(s_1, d_1) = o(s_2, d_1) = u_1$} and \mbox{$o(s_1,d_2) = o(s_2,d_2) = u_2$}. Let the controller parameter be \mbox{$\theta  = ( \overline{u}_1, \overline{u}_2) = (1,2)$}. Then, the parameterized output function is \mbox{$o(s_1, d_1; \theta) = o(s_2, d_1; \theta) = 1$} and \mbox{$o(s_1,d_2; \theta) = o(s_2,d_2; \theta) = 2$}. \END 

\begin{figure}[]
	\centering
	\includegraphics[width=0.2\textwidth]{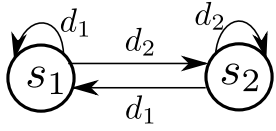}
	\caption{A toy mealy machine controller.}
	\label{fig:toy_fts}
\end{figure}
\end{example}

 Next, we provide guidance on how to construct $C_{cl}$ efficiently. Since $s\in Q$, with $Q$ finite, we can decompose $C_{cl}$ into $ \vert Q \vert $ subsets:
\begin{align}
	C_{cl} =  \bigcup_{s_{i}\in Q} C_{sub}(s_{i})\times \{s_{i}\} ,
\end{align}
where:
\begin{align} \label{eqn:C_sub} 
	C_{sub}(s_{i}) =  \left\{(x, \theta) \in \R^{n+mL} \mid \mathcal{R}( \Sigma_{cl}, (x, \theta, s_{i})) \subseteq S_{cl}\right\}. 
\end{align}
For each state $(x', \theta, s') \in \mathcal{R}(\Sigma_{cl}, (x,\theta, s_{i})) $:
\begin{align}
\label{eq:show_fin}
	(x', \theta, s') \in S_{cl} \Leftrightarrow (x', o(s',d; \theta))\in S, \ \forall d\in D.
\end{align}
As $x'$ and $ o(s',d; \theta)$ are linear functions of $x$, $\theta$, and $d\in D$, the condition $(x', o(s',d; \theta))\in S$ is a set of linear inequality constraints on $(x, \theta)$. Thus, by \eqref{eq:show_fin} and the fact that $\mathcal{R}(\Sigma_{cl}, (x, \theta, s_{i}))$ is finite, the set $C_{sub}(s_{i})$ 
in \eqref{eqn:C_sub} can be expressed by a set of linear inequality constraints, that is a polytope in $\R^{n+mL}$. We use the following example to illustrate the computation of $C_{sub}(s_{i})$.

\begin{example} 
\label{ex:c_sub} 
Consider the mealy machine controller in Example~\ref{ex:toy_dfts}. Suppose that the nilpotent matrix $A$ of system $\Sigma$ satisfies $A^{2}=0$. Then, the reachable set $\mathcal{R}(\Sigma_{cl}, (x, \theta, s_{i}))$  contains $7$ elements, that is
\begin{align*}
	\mathcal{R}(\Sigma_{cl}, (x, \theta, s_{i})) &=\left\{ \bigcup_{j,k=1,2} (ABu_{j} + Bu_k + Ad_j+  d_{k}, \theta, s_{k}) \right\}\\
		&\hspace*{-7.5mm}\cup\bigg\{ (x, \theta, s_i) \bigg\} \cup \left\{ \bigcup_{j=1,2} (Ax+B u_j+ d_j, \theta, s_{j}) \right\},
\end{align*}
where $ \theta = (u_1,u_2)$. Suppose that the safe set is \mbox{$S = \left\{ (x,u) \in \R^{n+m} \mid G_xx+ G_{u}u \leq h\right\}$}. Then:
\begin{align*}
	C_{sub}(s_{1}) = C_{sub}(s_2) &= \\
	    \{(x,u) \in \R^{n+m} \mid &G_{x} x +G_{u}u_{i}\leq h,\\
        &G_{x}(Ax+Bu_{j}+d_{j}) + G_{u}u_{j} \leq h, \\
        &G_{x}(ABu_{j} + Bu_k + Ad_j+  d_{k}) + G_{u}u_{i}\leq h, \\
        &\forall i,j,k\in\{1,2\}\}.
\end{align*}
Finally, $ C_{cl}  =  \bigcup_{i=1,2}C_{sub}(s_i)\times \{s_{i}\}.$ 
\end{example}

So far we constructed $C_{cl}$ in closed-form. However, to obtain an RCIS from $C_{cl}$, we have to project $C_{cl}$ onto the first $n$ coordinates and then compute the convex hull of the projected set. Both projection and convex hull operations are time consuming and thus undesirable. In what follows we derive an \emph{implicit expression} of the resulting RCISs. 

\subsection{Implicit Controlled Invariant Set Expression (Method 1)} \label{sec:implicit_cinv_1} 
\noindent \textbf{Assumption 4:} The safe set $S$ of $\Sigma$ is bounded. 

Given that $S$ is bounded, the projection of $C_{sub}(s_{i})$ onto the first $n$ coordinates is also bounded. Thus, we can always find a large enough hyperbox $B$ such that:
\begin{align*}
    \textrm{Proj}_{1:n}( B\cap C_{sub}(s_{i})) = \textrm{Proj}_{1:n}(C_{sub}(s_{i})).
\end{align*}
Denote the intersection of the hyperbox $B$ and the polytope $C_{sub}(s_{i})$ by $\overline{C}_{sub}(s_{i})= B \cap C_{sub}(s_{i})$. The projection of $C_{cl}$ is exactly the union of the projections of polytopes $\overline{C}_{sub}(s_{i})$ over $s_{i}\in Q$, that is:
\begin{align}
	\textrm{Proj}_{1:n}(C_{cl}) = \bigcup_{s_{i}\in Q} \textrm{Proj}_{1:n} \left( \overline{C}_{sub}(s_{i}) \right). 
\end{align}
Since the order of convex hull operation and the projection can be swapped, we have that:
\begin{align}
\label{eqn:ch_ex} 
    \begin{split}
	\mathbf{CH} (\textrm{Proj}_{1:n}(C_{cl})) 
	= \textrm{Proj}_{1:n} \left( \mathbf{CH} \left( \bigcup_{s_{i}\in Q} \overline{C}_{sub}(s_{i}) \right) \right).  
	\end{split}
\end{align}
Since $ \overline{C}_{sub}(s_{i})$ is a polytope, 
it can be written as:
\begin{align*}
	C_{sub}(s_{i}) = \{ (x, \theta) \mid G_{i}(x, \theta) \leq h_{i}\}.
\end{align*}
Then, we construct the polytope:
\begin{align}
\label{eqn:C_lambda}
\begin{split}
    	C_{\lambda} = \bigg\{  &\left(x, \theta, x_1, \theta_1, \cdots,x_{\vert Q\vert}, \theta_{\vert Q\vert}, \lambda_1, \cdots, \lambda_{ \vert Q \vert }\right) \mid \\
 & \lambda_{i} \geq 0, G_{i} (x_i, \theta_i ) \leq \lambda_{i}h_{i},\forall 1 \leq i \leq \vert Q\vert, \\
  & \sum^{ \vert Q \vert }_{i=1} \lambda_{i}= 1,  \sum^{ \vert Q \vert }_{i=1} x_i = x, \sum^{ \vert Q \vert }_{i=1} \theta_i = \theta\bigg\}. 
\end{split}
\end{align}
Under Assumption 4, given that $(x,\theta)\in R^{n+mL}$, we have:
\begin{align}
	\mathbf{CH} \left( \bigcup_{s_{i}\in Q} \overline{C}_{sub}(s_{i}) \right)  = \textrm{Proj}_{1:(n+mL)} \left(C_{ \lambda} \right) , \\
	\textrm{Proj}_{1:n}\left(\mathbf{CH} \left( \bigcup_{s_{i}\in Q} \overline{C}_{sub}(s_{i}) \right) \right)  = \textrm{Proj}_{1:n}(C_{ \lambda}). \label{eqn:ch_proj} 
\end{align}
By \eqref{eqn:ch_ex} and \eqref{eqn:ch_proj}, the RCIS $ \mathbf{CH}(\textrm{Proj}_{1:n}(C_{cl}))$ is the projection of $C_{ \lambda}$ onto the first $n$ coordinates. In other words, $C_{ \lambda}$ is an \emph{implicit expression} of the RCIS $ \mathbf{CH}(\textrm{Proj}_{1:n}(C_{cl})).$

\begin{remark}
	Assumption 4 is only required if we want the equality $ \textrm{Proj}_{1:n}(C_{ \lambda})= \mathbf{CH}(\textrm{Proj}_{1:n}(C_{cl}))$ to hold. In the next subsection, we introduce an alternative implicit expression which does not require a bounded safe set $S$.
\end{remark}

\subsection{Implicit Controlled Invariant Set Expression (Method 2)} \label{sec:implicit_cinv_2} 
In Example~\ref{ex:c_sub}, the projection of $C_{cl}$ is already convex and, thus, the convex hull computation is omitted. It turns out that the convexity of $\textrm{Proj}_{1:n}(C_{cl})$ is not a coincidence. We define the nested state transition function by:
\begin{align}
\hspace*{-1em}
\mathcal{T}^{*}(s, (d(t))_{t=0}^{k})
\hspace*{-0.2em}
=
\hspace*{-0.2em}
\begin{cases}
\hspace*{-0.2em}
	\mathcal{T}(s,d(0)), & k=0,\\
\hspace*{-0.2em}
	\mathcal{T}(\mathcal{T}^{*}(s, (d(t))_{t=0}^{k-1}), d(k)), & k>0.
\end{cases}
\end{align}
Similarly, the nested output function $o^{*}(s,(d(t))_{t=0}^{k})$ is:
\begin{align} 
\label{eqn:o_star} 
\hspace*{-0.65em}
o^{*}(s, (d(t))_{t=0}^{k})=
\begin{cases}
	o(s,d(0)), & k=0,\\
	o(\mathcal{T}^{*}(s, (d(t))_{t=0}^{k-1}), d(k)), & k>0.
\end{cases}
\end{align}

Define a preorder relation ``$ \succeq $" on $Q$ as follows. For any \mbox{$s_{1}, s_{2}\in Q$},  we have that $s_{1} \succeq  s_{2}$ if for all \mbox{$(d_{1}(t))_{t=0}^{k_1} \in D^{k_1}$} and \mbox{$(d_{2}(t))_{t=0}^{k_2} \in D^{k_2}$} with non-negative integers \mbox{$k_1, k_2 \leq  \vert Q \vert^{2}$}, \mbox{$o^{*}( s_{1}, (d_1(t))_{t=0}^{k_1}) = o^{*}(s_{1}, (d_2(t))_{t=0}^{k_2})$} implies  \mbox{$o^{*}( s_{2}, (d_1(t))_{t=0}^{k_1}) = o^{*}(s_{2}, (d_2(t))_{t=0}^{k_2})$} .

Here, the ``$=$" sign in \mbox{$o^{*}(s_{i}, (d_1(t))_{t=0}^{k_1}) = o^{*}(s_{i}, (d_2(t))_{t=0}^{k_2})$} is interpreted as the function $o^{*}$ mapping two inputs to the same element in $ \Theta$ (regardless of the parameter $ \theta $). Given the definition of the relation $ \succeq $ on $Q$, we can algorithmically check if two states $s$ and $ s' $ satisfy $ s\succeq s' $  with worst case time complexity $O(\vert Q \vert^{2})$.

Note that the ``$ \succeq $'' relation is not a partial order as it does not satisfy the antisymmetry condition, namely it is possible to have $s_1 \succeq s_2$ and $s_2 \succeq s_1$ but $s_1 \not=s_2$. However, the following theorem shows that the $\succeq $ relation in $Q$ actually implies the partial order on the sets $\{ \textrm{Proj}_{1:n}(C_{sub}(s))\}_{s\in Q}$ defined by the set inclusion.
\begin{theorem} \label{thm:partial} 
	Given the $ \succeq $ relation defined on the set $Q$, for any two states $s_1$, $s_2\in Q$, $s_1 \succeq s_2$ implies that \mbox{$\textrm{Proj}_{1:n}(C_{sub}(s_{1})) \supseteq \textrm{Proj}_{1:n} (C_{sub}(s_2))$}. 
\end{theorem}
We call a state $s_{max}\in Q$ a \emph{maximal state} if for any $ s'\in Q $, $ s'\succeq s_{max} $ implies $ s_{max} \succeq s' $, and call a state $ s_{dom} $ a \emph{dominant state} if $s_{dom} \succeq s$ for all $s\in Q$. Denote $Q_{max}$ as the set of all the maximal states in $Q$.

\begin{corollary} \label{cor:1} 
	Suppose that there exists a dominant state $s_{dom}\in Q$.
	Then, \mbox{$\textrm{Proj}_{1:n}(C_{cl}) = \textrm{Proj}_{1:n}(C_{sub}(s_{dom}))$}. 
\end{corollary}
Corollary~\ref{cor:1}  explains our observation in Example~\ref{ex:c_sub}.
\begin{example}
	For the mealy machine in Example~\ref{ex:toy_dfts}, $s_1 \succeq s_2$ and $s_2 \succeq s_1$.  Thus, both $s_1$ and $s_2$ are dominant states. Then, \mbox{$\textrm{Proj}_{1:n}(C_{cl}) =\textrm{Proj}_{1:n} (C_{sub}(s_1)) = \textrm{Proj}_{1:n}(C_{sub}(s_2))$}.
\end{example}
\begin{corollary} \label{cor:2} 
	Define a partition over $ Q_{max} $ as follows: For $ s $ and $ s'\in Q $, $ s $ and $ s' $ belong to the same component if $ s  \succeq s'$ and/or $ s' \succeq s $. Let a set $ Q_{0} \subseteq Q_{max} $ contain exactly one state from each component of this partition. Then, \mbox{$\textrm{Proj}_{1:n}(C_{cl}) =\cup_{s_{max}\in Q_{0}} \textrm{Proj}_{1:n} (C_{sub}(s_{max}))$}. 
\end{corollary}

According to Corollary~\ref{cor:1}, if a dominant state $s_{dom}$ exists in $ Q$, the RCIS $ \mathbf{CH}(\textrm{Proj}_{1:n}(C_{cl}))$ is simply the projection of $C_{sub}(s_{dom})$ onto the first $n$ coordinates. In this case, we can directly take $C_{sub}(s_{dom})$ as the implicit representation of the RCIS $ \mathbf{CH}(\textrm{Proj}_{1:n}(C_{cl}))$; otherwise, by Section \ref{sec:implicit_cinv_1}, we construct $ C_{\lambda} $ as  the implicit RCIS.  Note that according to Corollary \ref{cor:2}, we can replace $ Q $ by $ Q_{0} $ in the definition of $ C_{\lambda} $.
The overall procedure of computing implicit RCISs is summarized in Algorithm~\ref{alg:1}. 

\begin{algorithm}[t!]
\caption{Compute Implicit Controlled Invariant Set}
\label{alg:1}
\begin{algorithmic}
	\State {\bf{inputs:}} $\Sigma$, $S$, $\Sigma_{c}=(Q,D,\mathcal{T},\Theta, o)$. 
	\vspace{2mm}
	\If {a dominant state $s_{dom}\in Q$ exists}
	\State Compute $C_{sub}(s_{dom})$ as in \eqref{eqn:C_sub}.
	\State \Return $C_{sub}(s_{dom})$.
	\Else
	\For  {$s_{i}\in Q_{0} \subseteq Q$ }
	\State Compute $C_{sub}(s_{i})$ as in \eqref{eqn:C_sub}.
	\EndFor
	\State Compute $C_{ \lambda}$ as in \eqref{eqn:C_lambda}. 
	\State \Return $C_{ \lambda}$.
	\EndIf
	\vspace{2mm}
\end{algorithmic}
\end{algorithm}

\section{Bridge Measurable and Non-measurable Disturbances} \label{sec:bridge} 
In this section we prove a connection between measurable and non-measurable disturbances, which enables our method to compute RCISs for systems with any type of disturbances.

Suppose a system $\Sigma$ in \eqref{eq:dtls} has a non-measurable disturbance $d\in D$. We construct a system $\Sigma'$ with a measurable disturbance by adding a one-step delay:
\begin{align} 
\label{eq:dtls_2} 
    \hspace*{-1em}
    \Sigma' \hspace*{-0.25em} : \hspace*{-0.25em} \begin{bmatrix}
		x(t+1) \\ u(t+1)	
    \end{bmatrix} = 
	\begin{bmatrix}
		A & B\\
		0 & 0
    \end{bmatrix} 
    \hspace*{-0.25em}
    \begin{bmatrix}
	x(t)\\u(t) 
    \end{bmatrix} + \begin{bmatrix}
    0\\I_{m}
		\end{bmatrix} v(t) + \begin{bmatrix}
    I_{n} \\0	
	\end{bmatrix} d(t),
\end{align}
with state $(x,u)\in \R^{n+m}$, input $v\in \R^{m}$, a measurable disturbance $d\in D \subseteq \R^{n}$, and $I_{n}$, $I_{m}$ being the $n\times n$ and $m\times m$ identity matrices respectively. 

Let the safe set of $\Sigma$ be $S \subset \R^{n+m}$. We want to compute an RCIS for $\Sigma$ within $S$. The next theorem reveals that this can be achieved by computing an RCIS for $\Sigma'$ within $S \times \R^{m}$.
\begin{theorem} \label{thm:non_meas} 
	Given the systems $\Sigma$ in \eqref{eq:dtls} and $\Sigma'$ in \eqref{eq:dtls_2}, if $C'$ is an RCIS for $\Sigma'$ in $S \times \R^{m}$, the projection $\textrm{Proj}_{1:n}(C')$ of $C'$ onto the first $n$ coordinates is an RCIS for $\Sigma$ in $S$. 

If $C'$ is the maximal RCIS for $\Sigma'$ in $S\times \R^{m}$, then $\textrm{Proj}_{1:n}(C')$ is the maximal RCIS for $\Sigma$ in $S$.
\end{theorem}
Thanks to Theorem~\ref{thm:non_meas}, in terms of computing RCISs, any method designed for measurable disturbances can be applied to systems with non-measurable disturbances.  

\section{Case Study}
\label{sec:compEval}

\subsection{Mealy machines with dominant states}
\label{sec:dfts} 
We present two classes of mealy machines that contain at least one dominant state.  
\subsubsection{Simple Loop} Given an integer $L>0$, let $Q = \{s_{i}\}_{i=1}^{L} $ and $\Theta =\{u_{i}\}_{i=1}^{L}$. Define, for all $d\in D$, the state transition and  output functions as:
\begin{align}
	\mathcal{T}(s_{i},d) = \begin{cases}
		s_{i+1} & i < L,\\
		s_{1} & i = L.
   \end{cases},~
	o(s_{i},d) = \begin{cases}
		u_{i+1} & i < L,\\
		u_{1} & i = L.
   \end{cases} 
\end{align}
For such a structure, any $s\in Q$ is a dominant state. 
\subsubsection{Tree Structure} Suppose the cardinality of the disturbance set $ \vert D \vert  = K$.  Given an integer $L>0$, define \mbox{$N = (K^{L}-1)/(K-1)$}.  Let $Q = \{s_0\}\cup \bigcup_{i=1}^{L} D^{i}$. That is, $Q$ is the union of $s_0$ and all finite sequences of elements in $D$ with length less than or equal to $L$. We assign one output for each $s\in Q$ denoted by $u(s)$. Thus, $\Theta = \{u(s)\}_{s\in Q}$. The state transition function is defined as for all $d\in D$:
\begin{align}
	\mathcal{T}(s,d) = \begin{cases}
		d & s = s_0,\\
		sd & s\in D^{k}, k < L,\\
		s(2:L)d & s\in D^{L},
	\end{cases}
\end{align}
where $sd\in D^{k+1}$ denotes the concatenation of \mbox{$s\in D^{k}$} and \mbox{$d\in D$}, and \mbox{$s(2:L)d$} denotes the concatenation of the subsequence $s(2:L)$ of $s$ and $d\in D$. For instance, if \mbox{$s=d_1 d_2 \cdots d_{L}$}, then \mbox{$s(2:L)d = d_2 \cdots d_{L} d$}.

 For $L=K=2$, the state transition function is shown in Fig.~\ref{fig:tree}. We call this class of mealy machines \emph{tree structure} since the mealy machine transition graph, as shown in Fig.~\ref{fig:tree}, embeds a tree with $s_0$ the root node.

Given the state transition function, the output function is simply defined as:
\begin{align}
	o(s,d) = u(\mathcal{T}(s,d)).
\end{align}

For any tree-structure mealy machine, $s_0$ is the dominant state. Intuitively, the tree-structure mealy machine memorizes the past $L$ disturbance measurements and assigns a control input to each possible combination of the past $L$ disturbances. 
\begin{figure}[]
	\centering
	\includegraphics[width=0.25\textwidth]{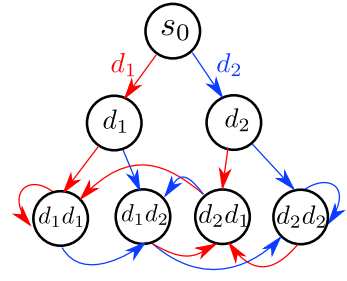}
	\caption{The tree-structure mealy machine ($L=3$, $D=\{d_1,d_2\}$). The red arrow and blue arrow indicate transitions under $d_1$ and $d_2$ respectively. }
	\label{fig:tree}
\end{figure}

Finally, for both classes of mealy machines introduced here, it can be proven that by increasing the number of discrete states (complexity), that is increasing $L$, we tend to obtain larger RCISs. 

\subsection{Lane keeping supervision}
Consider a $4$-dimensional linearized bicycle vehicle dynamics with respect to a constant longitudinal velocity $30m/s$ in \cite{smith2016interdependence}, discretized with time step $\Delta t = 0.1 s$. The system states consist of the lateral displacement $y$, lateral velocity $v$, yaw angle $\Delta \Psi$ and yaw rate $r$. The control input $u$ is the steering angle. The disturbance is $d = (0, 0, - r_{d}\Delta t, 0)$, where $r_{d}\in\R$ is the road curvature within a range $ \vert r_{d} \vert \leq r_{d,max}$.  The safe set is given by constraints $ \vert y \vert\leq 0.9$, $ \vert v \vert \leq 1.2$, $ \vert  \Delta \Psi \vert \leq  0.05 $, $ \vert r \vert \leq 0.3$ and $ \vert u\vert \leq \pi/2$. 

The future road curvature can be measured in ahead and thus $d$ is a measurable disturbance\cite{xu2019design}. 
We compare our method with Method 2 in \cite{anevlavis2021enhanced}, LMI-based low-complexity RCIS in \cite{tahir2014low} 
and the Maximal RCIS. Our method uses the tree structure with $L=4$ in Section~\ref{sec:dfts} as the mealy machine controller.  For Method 2 in  \cite{anevlavis2021enhanced}, we set the parameter $ L=14 $ and compute the lifted set in high dimensional space as an implicit RCIS. Note that Method 2 with parameter $ L$ is the same as our method equipped with simple loop controller with parameter $L$. For the LMI-based method in \cite{tahir2014low}, we set the parameter $ \rho=1 $ and run the iterative algorithm until convergence.  The methods in \cite{anevlavis2021enhanced}, \cite{tahir2014low} consider non-measurable disturbances only. To make a fair comparison, our method computes  RCISs  for $ d $ being measurable and/or non-measurable respectively.  We evaluate the algorithm performance by their computation time  and the volume percentage of the resulting RCISs to the Maximal RCIS. The volume percentage is estimated by monte carlo method with sample size $ N=10^4 $.

The comparison results are shown in Table \ref{tab:res_1}: According to the $2$nd, $ 3 $rd and $ 4 $th  rows of Table \ref{tab:res_1}, when dealing with non-measurable disturbances only, our method outperforms Method 2 of \cite{anevlavis2021enhanced} and LMI-based method in \cite{tahir2014low} in both the computation time and the volume of the resulting RCIS for all $ r_{d,max} $, showing a strong robustness to non-measurable disturbances. The LMI-based method encounters an infeasible optimization problem in all test cases and thus has $ 0 $ volume percentage. Method 2 of \cite{anevlavis2021enhanced}, as a special case of the proposed method, has a decent volume percentage when the disturbance range is small. But as $ r_{d,max} > 0.015 $, the RCIS from Method 2 of \cite{anevlavis2021enhanced} becomes empty, while our method still has volume percentage greater than $ 98\% $ for both measurable and non-measurable cases. 

Shown by the first  $ 2 $ rows of Table \ref{tab:res_1}, when $ r_{d,max} = 0.07$, our method returns a nonempty RCIS for $ d $ being measurable, but returns an empty RCIS for $ d $ being non-measurable. Thus, by considering $ d $ as a measurable disturbance, our method is robust to a larger range of disturbances. Finally, comparing the first $ 2 $ rows with the last row of Table \ref{tab:res_1}, when $ r_{d,max} < 0.07 $, our method computes implicit RCISs with almost the same size as the Maximal RCISs, using less than $ 0.3\% $ computation time of the Maximal RCISs. 

\begin{table}[]
	\centering
	\caption{Computation Time and Volume Percentage of Computed RCIS to the Maximal RCIS. (Lane Keeping)}
	\label{tab:res_1}
	\scriptsize
	\begin{tabular}{c|cccccc}
		\hline
		$r_{d,max}$ && $0.01$ & $0.015$ &  $0.03$ & $0.05$ & $0.07$\\ 
		\hline
  		\multirow{2}{6em}{Our method ($ d $ meas.)}& Time (s)& $0.042$ & $0.035$  & $0.037$ & $0.035$  & $0.032$\\
  		& Vol (\%)& $100.00$ & $99.99$ & $99.89$ &  $98.91$ & $74.75$ \\
  		\hline
		\multirow{2}{6em}{Our method ($ d $ non-meas.)}& Time (s)& $0.071$ & $0.062$  & $0.072$ & $0.063$  & $0.060$\\
  		& Vol (\%)& $100.00$ & $100.00$ & $100.00$ &  $99.89$ & $0$ \\
  		  		\hline
		\multirow{2}{6em}{Method 2 of \cite{anevlavis2021enhanced} ($L=14$)} & Time (s) & $0.506$ & $0.443$ & $0.404$ & $0.397$ & $0.484$  \\
	  	 &Vol (\%) & $99.82$ & $87.64$ & $ 0$ & $0$ & $0$\\
	  	   		\hline
		\multirow{2}{6em}{LMI Method \cite{tahir2014low} ($ \rho=1 $)} & Time (s) & $0.449$ & $0.519$ & $0.562$ & $0.500$ & $0.564$  \\
		&Vol (\%) & $0$ & $0$ & $ 0$ & $0$ & $0$\\
		  		\hline
		Maximal RCIS & Time (s) & $13.084$ & $18.918$ & $15.525$ & $15.698$ & $21.513$\\
		\hline
	\end{tabular}
\end{table}

Next, we illustrate how the computed implicit RCIS can be used to supervise a nominal controller. Suppose the current state $x(t)$ belongs to the RCIS $\textrm{Proj}_{1:n}(C_{sub}(s_0))$.  Given the nominal steering input $u_{d}(t)$ and the disturbance $d(t)$ at time $t$, we minimally change the input $u_{d}(t)$ such that the next state $x(t+1)$ stays in the RCIS $\textrm{Proj}_{1:n}(C_{sub}(s_0))$ by solving the following quadratic program:
\begin{align}
	\begin{split} \label{eqn:opt_project} 
		\min_{u(t), \theta}& \Vert u_{d}(t) -u(t)\Vert_{2}^{2}\\
		\text{ subject to }& (Ax(t) + Bu(t)+ d(t), \theta)\in C_{sub}(s_0)  , 
	\end{split}
\end{align}
where $A$, $B$ are the system matrices and $ \theta$ is a slack variable.  
We use the solution $u(t)$ of \eqref{eqn:opt_project} as the actual steering input to the vehicle.   The feasibility of \eqref{eqn:opt_project} is guaranteed since $\textrm{Proj}_{1:n}(C_{sub}(s_0))$ is an RCIS and $x(t)\in \textrm{Proj}_{1:n}(C_{sub}(s_0))$. 

We compare the supervised inputs obtained in \eqref{eqn:opt_project} to the ones obtained based on the Maximal RCIS $C_{max}$ via the following quadratic program:
\begin{align}
	\begin{split} \label{eqn:opt_max} 
		\min_{u(t)}& \Vert u_{d}(t) -u(t)\Vert_{2}^{2}\\
		\text{ subject to }& Ax(t) + Bu(t)+ d(t)\in C_{max}.
	\end{split}
\end{align}

The nominal controller is $u_{d}(t) = -0.1812 y - 0.0373 v - 4.5996 \Delta\Psi - 0.6649 r$. We run two simulations with the same initial states and control inputs obtained from \eqref{eqn:opt_project} and \eqref{eqn:opt_max} respectively ($r_{d,max}=0.015$). As shown in Fig.~\ref{fig:road} and Fig.~\ref{fig:input}, the vehicle maneuvers and steering inputs supervised by our implicit RCIS $C_{sub}(s_0)$ and the Maximal RCIS are very close to each other. The maximal difference between the control inputs from  \eqref{eqn:opt_project} and \eqref{eqn:opt_max} is around $ 0.035 $ at $ t=6.5s $.  This observation is consistent with the results shown in Table~\ref{tab:res_1}, where the volume of our implicit RCIS is approximately $99.96\%$ of the volume of the Maximal RCIS.

\begin{figure}[]
	\centering
	\includegraphics[width=0.45\textwidth]{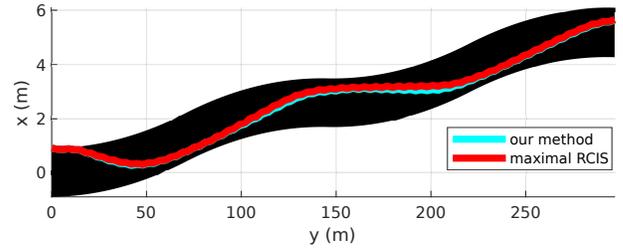}
	\caption{Vehicle maneuvers under control inputs supervised by our implicit RCIS (cyan curve) and the Maximal RCIS (red curve). The black region indicates the road surface.}
	\label{fig:road}
\end{figure}

\begin{figure}[]
	\centering
	\includegraphics[width=0.5\textwidth]{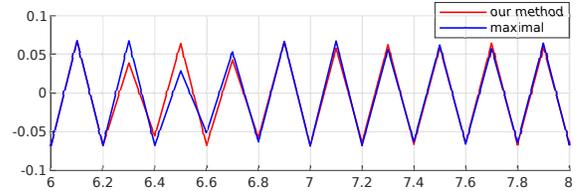}
	\caption{Vehicle steering inputs supervised by our implicit RCIS (red curve) and the Maximal RCIS (blue curve) ($ 6 \leq t\leq 8 $).}
	\label{fig:input}
\end{figure}

\subsection{Chain of integrators}
Consider a discrete-time $ n $-th order integrator:
\begin{align}
x(t+1)= \left( I_{n} + \begin{bmatrix}
0 & I_{n-1}\\
0 & 0
\end{bmatrix} \right) x(t) + \begin{bmatrix}
0\\1	
\end{bmatrix} (u(t)+ d(t))
\end{align}
with $ x\in \R^{n} $, $u\in\R $ and $ d\in \R$. $I_{n}$ indicates the identity matrix in $\R^{n\times n}$. $d$ is considered as a measurable disturbance within range $\vert d\vert \leq 0.1 $. The safe set is \mbox{$S = \{(x,u) \mid \vert x_{i} \vert \leq 1, \forall i=1, \dots, n,~ \vert u \vert \leq 1\}.$}

\begin{figure}[]
	\centering
	\includegraphics[width=0.3\textwidth]{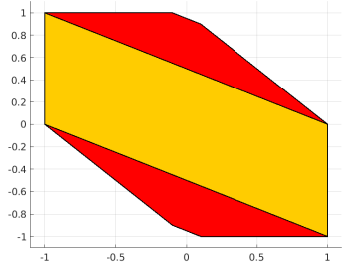}
	\caption{RCISs for double integrator. Yellow: RCIS from LMI-based method \cite{tahir2014low}. Red: the Maximal RCIS computed by both our method and \cite{anevlavis2021enhanced}.}
	\label{fig:double_integrator}
\end{figure}
 The comparison results of our approach (tree structure, $L=4$) with Method 2 of \cite{anevlavis2021enhanced} ($L=14$) and the LMI-based method in \cite{tahir2014low} ($\rho=1$) are shown in Table \ref{tab:res_3}. For $n \leq 4$, our method outperforms the other 2 methods in computation time and volume percentage.  For $n=2$, our method returns exactly the Maximal RCIS, depicted in Fig.  \ref{fig:double_integrator}. For $n\geq 6$, the Maximal RCIS does not terminate within $1$ hour. Thus we only check if the computed RCISs are empty or not instead of comparing their volume to the Maximal RCIS. When $n\geq 6$, our method is the only one that returns non-empty RCISs. Note that even though the implicit RCIS has closed-form expression, the number of constraints in the implicit RCIS grows exponentially as $n$ increases. In this example, for $n=10$, it takes about $339$s for our method to generate the implicit RCIS, which is a polytope in $\R^{24}$ with about $36\times 10^4$ constraints.
\begin{table}[]
	\centering
	\caption{Computation Time and Volume Percentage of Computed RCIS to the Maximal RCIS (Chain of Integrators).}
	\label{tab:res_3}
	\scriptsize
	\begin{tabular}{c|cccccc}
		\hline
		$n$ && $2$ & $4$ &  $6$ & $8$ & $10$\\ 
		\hline
  		\multirow{2}{6em}{Our method ($ d $ meas.)}& Time (s)& $0.005$ & $0.025$  & $0.368$ & $9.287$  & $339.060$\\
  		& Vol (\%)& $100$ & $98.79$ & $>0$ &  $>0$ & $>0$ \\
  		\hline
		\multirow{2}{6em}{Method 2 of \cite{anevlavis2021enhanced} ($L=14$)} & Time (s) & $0.183$ & $0.341$ & $2.420$ & $7.168$ & $37.105$  \\
	  	 &Vol (\%) & $74.49$ & $0$ & $ 0$ & $0$ & $0$\\
	  	   		\hline
		\multirow{2}{6em}{LMI Method \cite{tahir2014low} ($ \rho=1 $)} & Time (s) & $3.465$ & $0.603$ & $0.952$ & $1.405$ & $3.1402$  \\
		&Vol (\%) & $66.85$ & $\approx 0$ & $ 0$ & $0$ & $0$\\
		  		\hline
		Maximal RCIS  & Time (s) & $0.734$ & $11.114$ & $>3600$ & $>3600$ & $>3600$\\
		\hline
	\end{tabular}
\end{table}

\subsection{Truck with $N$ trailers}
Consider a continuous-time model for a truck with $N$ trailers \cite{rungger2013guidedsynthltl}. The state consists of the $N+1$ velocity values, each for the truck and the $N$ trailers, and the $N$ spring elongations in between them. Hence, $N$ trailers correspond to dimension $n = 2N+1$. The input is the velocity of the truck. We discretize the model with a sampling time of $T_s$ seconds assuming piecewise constant inputs.

\begin{table}[t!]
\caption{Computation  Time  and  Volume  Percentage  of Computed RCIS to the Maximal RCIS (Truck with $N$ trailers).}
\label{table:truck}
\centering
\scriptsize
\begin{tabular}{c|ccccc}
\hline
System dimension & & $n=3$ & $n=5$ & $n=7$ & $n=9$   \\
\hline

\multirow{2}{6em}{Our method}
								& Time (s.)   & $0.109$ & $0.781$ & $8.669$  & $163.1$ \\
								& Vol ($\%$)  & $100$ & $100$ & $>0$  & $0$ \\				
\hline

\multirow{2}{6em}{Method 2 of \cite{anevlavis2021enhanced} (L=14)}
								& Time (s.)   & $0.547$ & $0.814$ & $1.352$  & 6.577 \\
								& Vol ($\%$)  & $100$ & $98.90$ & $0$  & $0$ \\				
\hline

Maximal RCIS                    & Time (s.)   & $0.746$ & $13.76$ & $>3600$ & $>3600$ \\ 
\hline
\end{tabular}
\vspace{-5mm}
\end{table}

Table~\ref{table:truck} shows the results of this case study for our method and the approach in \cite{anevlavis2021enhanced}. For $n\geq7$ the method computing the Maximal RCIS does not terminate after 1 hour, and, hence, we only check non-emptiness of sets instead of volume percentage. When the Maximal RCIS is computed, we see that our approach covers it, but due to the implicit representation, the running times are much faster. However, we see that in this example, after some point, as the dimension becomes large, the set our algorithm returns is empty. This can be understood as by adding more trailers the noise from each spring compounds towards the ones behind it, resulting in the shrinking of the RCIS.  

\section{Conclusion}\label{sec:concl}
In this paper, we present a novel method of computing implicit RCISs in closed-form. The key insight is to construct a closed-loop system with a parameterized automaton-based controller. The implicit RCISs obtained by our method characterize the set of feasible initial states and controller parameters under which the system state-input trajectory stays in the safe set. Compared with the standard iterative methods\cite{bertsekas1972infreach, nilsson2017correct}, all the computations of our method are done in one-shot, which guarantees finite-time termination and better scalability. Several numerical examples are provided to demonstrate the efficiency and practicality of the proposed method.

\bibliographystyle{IEEEtran}
\bibliography{ref}

\begin{thebibliography}{10}
\providecommand{\url}[1]{#1}
\csname url@samestyle\endcsname
\providecommand{\newblock}{\relax}
\providecommand{\bibinfo}[2]{#2}
\providecommand{\BIBentrySTDinterwordspacing}{\spaceskip=0pt\relax}
\providecommand{\BIBentryALTinterwordstretchfactor}{4}
\providecommand{\BIBentryALTinterwordspacing}{\spaceskip=\fontdimen2\font plus
\BIBentryALTinterwordstretchfactor\fontdimen3\font minus
  \fontdimen4\font\relax}
\providecommand{\BIBforeignlanguage}[2]{{%
\expandafter\ifx\csname l@#1\endcsname\relax
\typeout{** WARNING: IEEEtran.bst: No hyphenation pattern has been}%
\typeout{** loaded for the language `#1'. Using the pattern for}%
\typeout{** the default language instead.}%
\else
\language=\csname l@#1\endcsname
\fi
#2}}
\providecommand{\BIBdecl}{\relax}
\BIBdecl

\bibitem{bertsekas1972infreach}
D.~Bertsekas, ``Infinite time reachability of state-space regions by using
  feedback control,'' \emph{Automatic Control, IEEE Transactions on}, vol.
  AC-17, pp. 604 -- 613, 11 1972.

\bibitem{nilsson2017correct}
L.~P. Nilsson, ``Correct-by-construction control synthesis for high-dimensional
  systems,'' Ph.D. dissertation, 2017.

\bibitem{rungger2017rcislinsys}
M.~{Rungger} and P.~{Tabuada}, ``Computing robust controlled invariant sets of
  linear systems,'' \emph{IEEE Transactions on Automatic Control}, vol.~62,
  no.~7, pp. 3665--3670, July 2017.

\bibitem{tahir2014low}
F.~Tahir and I.~M. Jaimoukha, ``Low-complexity polytopic invariant sets for
  linear systems subject to norm-bounded uncertainty,'' \emph{IEEE Transactions
  on Automatic Control}, vol.~60, no.~5, pp. 1416--1421, 2014.

\bibitem{gupta2018full}
A.~Gupta and P.~Falcone, ``Full-complexity characterization of
  control-invariant domains for systems with uncertain parameter dependence,''
  \emph{IEEE control systems letters}, vol.~3, no.~1, pp. 19--24, 2018.

\bibitem{michel2018invariant}
N.~Michel, S.~Olaru, G.~Valmorbida, S.~Bertrand, and D.~Dumur, ``Invariant sets
  for discrete-time constrained linear systems using a sliding mode approach,''
  in \emph{2018 European Control Conference (ECC)}.\hskip 1em plus 0.5em minus
  0.4em\relax IEEE, 2018, pp. 2929--2934.

\bibitem{anevlavis2019cis2m}
\BIBentryALTinterwordspacing
T.~{Anevlavis} and P.~{Tabuada}, ``Computing controlled invariant sets in two
  moves,'' in \emph{2019 IEEE 58th Conference on Decision and Control (CDC)},
  2019, pp. 6248--6254. [Online]. Available:
  \url{https://doi.org/10.1109/CDC40024.2019.9029610}
\BIBentrySTDinterwordspacing

\bibitem{anevlavis2020simple}
\BIBentryALTinterwordspacing
T.~Anevlavis and P.~Tabuada, ``A simple hierarchy for computing controlled
  invariant sets,'' in \emph{Proceedings of the 23rd International Conference
  on Hybrid Systems: Computation and Control}, ser. HSCC '20.\hskip 1em plus
  0.5em minus 0.4em\relax New York, NY, USA: Association for Computing
  Machinery, 2020. [Online]. Available:
  \url{https://doi.org/10.1145/3365365.3382205}
\BIBentrySTDinterwordspacing

\bibitem{anevlavis2021enhanced}
T.~Anevlavis, Z.~Liu, N.~Ozay, and P.~Tabuada, ``An enhanced hierarchy for
  (robust) controlled invariance,'' in \emph{2021 American Control Conference
  (ACC)}.\hskip 1em plus 0.5em minus 0.4em\relax IEEE, 2021, pp. 4860--4865.

\bibitem{wintenberg2020implicit}
A.~Wintenberg and N.~Ozay, ``Implicit invariant sets for high-dimensional
  switched affine systems,'' in \emph{2020 59th IEEE Conference on Decision and
  Control (CDC)}.\hskip 1em plus 0.5em minus 0.4em\relax IEEE, 2020, pp.
  3291--3297.

\bibitem{anevlavis2021controlled}
T.~Anevlavis, Z.~Liu, N.~Ozay, and P.~Tabuada, ``Controlled invariant sets:
  implicit closed-form representations and applications,'' \emph{arXiv preprint
  arXiv:2107.08566}, 2021.

\bibitem{liu2021safe}
Z.~Liu and N.~Ozay, ``Safe online planning in unknown nonconvex environments
  with implicit controlled invariant sets,'' \emph{IFAC-PapersOnLine}, vol.~54,
  no.~5, pp. 163--168, 2021.

\bibitem{xu2019design}
S.~Xu and H.~Peng, ``Design, analysis, and experiments of preview path tracking
  control for autonomous vehicles,'' \emph{IEEE Transactions on Intelligent
  Transportation Systems}, vol.~21, no.~1, pp. 48--58, 2019.

\bibitem{Antsaklis2006LinSys}
P.~J. Antsaklis and A.~Michel, \emph{Linear Systems}, 1st~ed.\hskip 1em plus
  0.5em minus 0.4em\relax Birkh\"auser Basel, 2006.

\bibitem{smith2016interdependence}
S.~W. Smith, P.~Nilsson, and N.~Ozay, ``Interdependence quantification for
  compositional control synthesis with an application in vehicle safety
  systems,'' in \emph{2016 IEEE 55th Conference on Decision and Control
  (CDC)}.\hskip 1em plus 0.5em minus 0.4em\relax IEEE, 2016, pp. 5700--5707.

\bibitem{rungger2013guidedsynthltl}
M.~Rungger, M.~Mazo, Jr., and P.~Tabuada, ``Specification-guided controller
  synthesis for linear systems and safe linear-time temporal logic,'' in
  \emph{Proceedings of the 16th International Conference on Hybrid Systems:
  Computation and Control}, ser. HSCC '13.\hskip 1em plus 0.5em minus
  0.4em\relax New York, NY, USA: ACM, 2013, pp. 333--342.

\end{thebibliography}

\balance

\appendix

\begin{proof} [Proof of Proposition~\ref{thm:RInv} ]
	Let $x\in C$ and $d\in D$. By the Definition~\ref{def:reachautonsys}, $ \mathcal{R}(\Sigma, f(x,d)) \subseteq \mathcal{R} (\Sigma_{a}, x) \subseteq S$ and thus $f(x,d) \in C$. Hence, $C$ is an RPIS by definition. 
	
	Suppose $x$ belongs to an arbitrary RPIS within  $S$. By definition, $ \mathcal{R}(\Sigma_{a},x) \subseteq S$. Thus, $x\in C$ and $C$ is the Maximal RPIS. 
\end{proof}

\begin{proof}[Proof of Proposition~\ref{thm:finite_D}]
	The "only if" direction is obvious. It is left to show the "if" direction. Suppose the safe set is $S$. Let $C$ be an RCIS for the system $\Sigma'$ in $S$ and $x$ be a point in $C$. We want to show that for all $d\in D$, there exists $u$ such that $Ax+Bu+d\in S$. Since $D = \mathbf{CH}(D_{v})$, there exists a finite $K>0$ such that  $d = \sum_{i=1}^{K} \alpha_{i} d_{i} $ for some $d_1$, ..., $d_{K}\in D_{v}$ and some $ \alpha_1$, ..., $ \alpha_{K} \geq 0$ satisfying $\sum_{i=1}^{K} \alpha_{i}=1$. Since $C$ is controlled invariant for $\Sigma'$, for each $d_{i}\in D_{v}$, there exists $u_{i}$ such that $(x,u_{i})\in S$ and $Ax+Bu_{i}+E d_{i}\in C$. Define $u = \sum^{K}_{i=1} \alpha_{i} u_{i}$. It is easy to show that $(x,u)\in S$ and $Ax +Bu+Ed \in C$, by the convexity of $S$ and $C$. Thus, $C$ is an RCIS within $S$ for $\Sigma$.
\end{proof}

\begin{proof}[Proof of Theorem~\ref{thm:CH}]
	Denote $ \mathbf{CH}(\textrm{Proj}_{1:n}(C_{cl}))$ by $C_{p}$. Let $x\in C_{p}$ and $d\in D$. We want to show that there exists $u$ such that $(x,u)\in S$ and $Ax+Bu+d\in C_{p}$.

	By definition of convex hull, there exist a positive integer $k > 0$, $k$ vectors $x_i\in \textrm{Proj}_{1:n}(C_{cl})$ and $k$ scalars $\alpha_{i}\in [0,1]$ for $i$ from $1$ to $k$ such that $ \sum^{k}_{i=1} \alpha_{i}= 1$ and $ \sum^{k}_{i=1} \alpha_{i}x_{i} = x$. For each $i$, there exists $\theta_{i}$ and $s_{i}$ such that $(x_{i},\theta_{i},s_{i})\in C_{cl}$. We define $u_{i} = o(s_{i},d)$. Note that by the definition of $S_{cl}$, $(x_{i},u_{i})\in S$. Also, since $C_{cl}$ is an RPIS, $(Ax_{i}+Bu_{i}+d, \theta_{i}, \mathcal{T}(s_{i},d))\in C_{cl}$ and thus $Ax_{i} + Bu_{i} + d\in C_{proj}$. We define $u = \sum^{k}_{i=1} \alpha_{i} u_{i}$. Since $S$ is convex and $(x_{i},u_{i})\in S$, $(x,u) = \sum^{k}_{i=1} \alpha_{i} (x_{i},u_{i}) \in S$. Since $C_{p}$ is convex and $Ax_{i}+Bu_{i}+d\in C_{p}$, $ Ax + Bu +d = \sum^{k}_{i=1} \alpha_{i}(Ax_{i}+Bu_{i}+d)\in C_{p}$.  Thus, $C_{p}$ is an RCIS for the system $\Sigma$ in $S$.
\end{proof}

\begin{proof}[Proof of Theorem~\ref{thm:finite_reach}]
	We want to show $ \mathcal{R}(\Sigma_{cl}, (x, \theta, s))$ is finite. Let $((x(t), \theta(t),s(t)))_{t=0}^{ \infty}$ be the trajectory of $\Sigma_{cl}$ with initial state $ (x(0), \theta(0), s(0)) = (x, \theta, s)$. Let $(d(t))_{t=0}^{ \infty}$ be the disturbance sequence. Given $A$ is nilpotent, that is $A^{h} = 0$ for some $h \geq 0$, we have that 
	\begin{align}
		\label{eqn:x} 
		x(t) = 
		\begin{cases}
			A^{t}x + \sum^{t-1}_{i=0}A^{t-1-i}[B o(s(i),d(i); \theta) + d(i)] & t < h,\\
			\sum^{t-1}_{i=t-h}A^{t-1-i}[B o(s(i),d(i); \theta) + d(i)] & t \geq  h.
		\end{cases}
	\end{align}
	Since $s(t)$ and $d(t)$ belong to finite sets $Q$ and $D$, $o(s(t),d(t); \theta)$ belongs to the finite set $ U( \theta) = \{o(s,d; \theta)\}_{s\in Q, d\in D}$.  Thus, according to \eqref{eqn:x},  $x(t)$, as a function of $o(s(t),d(t); \theta)$ and $d(t)$ for $t\geq h$, must belongs to a finite set, denoted by $X(\theta)$.  Thus, the reachable set $\mathcal{R}(\Sigma_{cl}, (x, \theta, s)) \subseteq X(\theta)\times \{ \theta\}\times Q $ is a finite set. 
\end{proof}

\begin{proof}[Proof of Theorem~\ref{thm:partial}]
	We want to derive a sufficient condition under which $\textrm{Proj}_{1:n}(C_{sub}(s_1)) \supseteq \textrm{Proj}_{1:n}(C_{sub}(s_2)) $. 
	Note that if for all $(x, \theta_2)\in C_{sub}(s_2)$, there exists $ \theta_1$ such that $(x, \theta_1)\in C_{sub}(s_1)$, then we have $\textrm{Proj}_{1:n}(C_{sub}(s_1)) \supseteq \textrm{Proj}_{1:n}(C_{sub}(s_2))$.

	Similar to how we define $o^{*}(s,d)$, we define the parameterized nested output function as
\begin{align} \label{eqn:o_star_p} 
o^{*}(s, (d(t))_{t=0}^{k}; \theta)=
\begin{cases}
	o(s,d(0); \theta) & k=0,\\
	o(\mathcal{T}^{*}(s, (d(t))_{t=0}^{k-1}), d(k); \theta) & k>0.
\end{cases}
\end{align}
	Given $s$ and $ \theta$, the parameterized nested output function $o^{*}(s, \cdot; \theta)$ becomes a function of $(d(t))_{t=0}^{k}$ in $\cup_{i=1}^{ \infty} D^{i}$. If for any $\theta_2$,  we can always find a $ \theta_1$ such that the functions $o^{*}(s_2, \cdot; \theta_2)= o^{*}(s_1, \cdot; \theta_1)$, then for all $(x, \theta_2)\in C_{sub}(s_2)$, $(x, \theta_1)\in C_{sub}(s_1)$. Intuitively, recall that $(x, \theta_2)\in C_{sub}(s_2)$ if $(x(k), o^{*}(s_2, (d(t))_{t=0}^{k}))\in S$ for all $k \geq 0$ and $(d(t))_{t=0}^{k}\in D^{k}$. If we know that $(x(k), o^{*}(s_2, (d(t))_{t=0}^{k}; \theta_2))\in S$ for all $k \geq 0$, then we know $(x(k), o^{*}(s_2, (d(t))_{t=0}^{k}; \theta_2))\in S$ for all $k \geq 0$ since $o^{*}(s_2,\cdot; \theta_2) = o^{*}(s_1,\cdot; \theta_1)$. Thus, $(x, \theta_1)\in C_{sub}(s_1)$.

	Now our goal is to derive a sufficient condition under which there exists a $ \theta_1$ such  that $ o^{*}(s_1, \cdot; \theta_1) = o^{*}(s_2, \cdot; \theta_2)$ for all $ \theta_2$. 

	\begin{lemma} \label{lemma:1} 
		Given $s_1$, $s_2\in Q$ and $ \theta_1$, $ \theta_2$, the functions $o^{*}(s_1,\cdot; \theta_1) = o^{*}(s_2,\cdot; \theta_2)$ if and only if $o^{*}(s_1, (d(t))_{t=0}^{k}; \theta_1) = o^{*}(s_2, (d(t))_{t=0}^{k}; \theta_2)$ for all $k \leq \vert Q \vert^{2}$ and all $(d(t))_{t=0}^{k}\in D^{k}$.
	\end{lemma}
	According to Lemma~\ref{lemma:1}, given any $ \theta_2$, we can directly solve for  a $ \theta_1$ satisfying $o^{*}(s_1, (d(t))_{t=0}^{k}; \theta_1) = o^{*}(s_2, (d(t))_{t=0}^{k}; \theta_2)$ for all $k \leq \vert Q \vert^{2}$ and all $(d(t))_{t=0}^{k}\in D^{k}$, which is a  system of linear equations on $ \theta_1$. It can be checked that given any $ \theta_2$, the solvability of the system of equations on $ \theta_1$ is guaranteed if for all $(d_1(t))_{t=0}^{k_1}$ and $(d_2(t))_{t=0}^{k_2}$ with $k_1$, $k_2\leq \vert Q \vert^{2}$, $o^{*}(s_1, (d_1(t))_{t=0}^{k_1}) = o^{*}(s_1, (d_2(t))_{t=0}^{k_2})$ implies $o^{*}(s_2, (d_1(t))_{t=0}^{k_1}) = o^{*}(s_2, (d_2(t))_{t=0}^{k_2})$, that is $s_1 \succeq s_2$ by definition.   
\end{proof}

\begin{proof}[Proof of Lemma~\ref{lemma:1} ]
	Given the mealy machine $(Q, D, \mathcal{T}, \Theta, o)$, we can construct a product mealy machine $(Q\times Q, D, \mathcal{T}_{pd}, \Theta\times \Theta,o_{pd})$ where for all $s_i $, $s_{j}\in Q $ and $ d\in D$
	\begin{align}
		\mathcal{T}_{pd}((s_{i},s_{j}), d) &=  (\mathcal{T}(s_i, d), \mathcal{T}(s_{j},d)), \\ 
		o_{pd}((s_{i},s_{j}), d) &= (o(s_{i},d), o(s_{j},d)). 
	\end{align}
	Given $ \theta_1$ and $ \theta_2$ as two value assignments of $\Theta$, we define the parameterized output function $o_{pd}((s_{i},s_{j}), d; \theta_1, \theta_2) = (o( s_{i}, d; \theta_1), o(s_{j}, d; \theta_2))$. 

	Given $s_1$, $s_2\in Q$ and $ \theta_1$ and $ \theta_2$, by construction, $o_{pd}^{*}((s_1,s_2), \cdot;\theta)$ is equal to $(o^{*}(s_1, \cdot; \theta_1),o^{*}(s_2, \cdot; \theta_2))$. Thus, $o^{*}(s_1, \cdot; \theta_1)\not=o^{*}(s_2, \cdot; \theta_2)$  if and only if there exists a $(d(t))_{t=0}^{k}$ such that $(s_{1}', s_{2}')= \mathcal{T}_{pd}^{*}((s_1,s_2), (d(t))_{t=0}^{k-1})$ and $o_{pd}((s_1',s_2'), d(k); \theta_1, \theta_2) = (u_1,u_2) $ for some $u_1$, $u_2\in \Theta$,  $\overline{u}_1\not= \overline{u}_2$. Since there are only $ \vert Q \vert^{2}$ states in the product mealy machine, if $(s_1',s_2')$ can be visited from $(s_1,s_2)$  under action sequence $(d(t))_{t=0}^{k-1}$, the smallest $k$ we need is less than or equal to $ \vert Q \vert^{2}$. Thus, if  $o^{*}(s_1, \cdot; \theta_1)\not=o^{*}(s_2, \cdot; \theta_2)$, there must exists a $(d(t))_{t=0}^{k}$ with $k \leq  \vert Q \vert^{2}$ such that $o^{*}(s_1, (d(t))_{t=0}^{k}; \theta_1)\not=o^{*}(s_2, (d(t))_{t=0}^{k}; \theta_2)$
\end{proof}

\begin{proof}[Proof of Theorem~\ref{thm:non_meas}]
	Denote $C_{p} = \textrm{Proj}_{1:n}(C')$. Suppose $C'$ is an RCIS of $\Sigma'$ in $S\times \R^{m}$. Let $x\in C_{p}$. We want to show that there exist $u$ such that $(x,u)\in S$ and for all $d\in D$, $Ax+Bu+d \in C_{p}$.  

	By definition of $C_{p}$, there exists $u\in \R^{m}$ such that $(x,u)\in C' \subseteq S$. Furthermore, since $C'$ is controlled invariant, there exists $v\in \R^{m}$ such that $(Ax+Bu+d, v)\in C'$ for all $d\in D$. Thus, $Ax+Bu+d\in C_{p}$ for all $d\in D$. Thus, we showed that $C_{p}$ is an RCIS for the system $\Sigma$ in $S$.

	Next, suppose that $C'$ is the Maximal RCIS for $\Sigma'$ in $S\times \R^{m}$. Also, suppose that $C_{max}$ is the Maximal RCIS for $\Sigma$ in $S$. We want to show that $\textrm{Proj}_{1:n}(C') = C_{max}$. Note that $\textrm{Proj}_{1:n}(C') \subseteq C_{max}$ as $\textrm{Proj}_{1:n} (C')$ is controlled invariant for $\Sigma$ in $S$. We need to show that $\textrm{Proj}_{1:n}(C') \supseteq C_{max} $, which is done in $3$ steps.

	First, define the set $C_{max}' = \{ (x,u) \mid (x,u)\in S, Ax+Bu+d\in C_{max}, \forall d\in D\}$. We want to show that $C_{max}'$ is controlled invariant for $\Sigma'$ in $S\times \R^{m}$.  Let $(x,u)\in C_{max}'$ and $d\in D$. By construction, $(x,u)\in S$ and $x^{+}= Ax+Bu+d\in C_{max}$. Since $C_{max}$  is controlled invariant for $\Sigma$, there exists $v\in \R^{m}$ such that $(x^{+},v)\in S$ and $Ax^{+} + Bu+ d^{+} \in C_{max}$ for all $d^{+}\in D$. Thus, by definition of $C_{max}'$, $(x^{+},v) = (Ax+Bu+d, v)\in C_{max}'$. Thus, $C_{max}' $ is an RCIS for $\Sigma'$ in $S\times \R^{m}$. 

	Second, as $C'$ is the Maximal RCIS for $\Sigma'$ in $S\times\R^{m}$, $C'\geq C_{max}'$. Thus, $\textrm{Proj}_{1:n}(C') \supseteq \textrm{Proj}_{1:n}(C_{max}')$. 

	Finally, note that for all $x\in C_{max}$, there exists $u$ such that $(x,u)\in S$ and $Ax+Bu+d\in C_{max}$ for all $d\in D$, namely that $(x,u)\in C_{max}'$. Hence, $C_{max} \subseteq \textrm{Proj}_{1:n}(C_{max}') \subseteq \textrm{Proj}_{1:n}(C')$. That is, $\textrm{Proj}_{1:n}(C')=C_{max}$ is the Maximal RCIS for $\Sigma$ in S. 
\end{proof}
\end{document}